\documentclass[a4paper]{article}

\usepackage{amsthm} 
\usepackage{amsmath} 
\usepackage{amssymb} 
\usepackage[ colorlinks=true
           , linkcolor=blue
           , citecolor=red
           , urlcolor=magenta]
           {hyperref}
\usepackage{geometry} 

\usepackage{stmaryrd} 
\usepackage{xspace}
\usepackage{soul}
\usepackage{cancel}

\usepackage{textgreek} 

\usepackage{multirow}
\usepackage{bussproofs}

\usepackage{tikz}
\usetikzlibrary{shapes}
\usetikzlibrary{arrows}
\usetikzlibrary{positioning}
\usetikzlibrary{decorations.pathreplacing}
\usetikzlibrary{matrix}
\usetikzlibrary{decorations.pathreplacing,calligraphy}
\usetikzlibrary{fit, positioning} 

\usepackage{subcaption}
\usepackage{graphicx}

\usepackage{framed}
\usepackage{wrapfig}
\usepackage[misc,geometry]{ifsym}

\usepackage{listings}
\newcommand{\listingsttfamily}{\fontfamily{SourceCodePro-TLF}\small}

\lstdefinestyle{c}{
    language=C,
    morekeywords={assert, true, false, then, fi, FOR, ROF, DIV0, DIV1, DIV, PRED0, PRED1, P0, P1, MUL1, MUL0, SUCC0, SUCC1, S0 ,S1, FOR01, ROF01, SAFE, IF, IF_ERR, ELSE, SKIP, CHECKPOINT, ROLLBACK, NORMAL, PUSH, POP},  
    backgroundcolor=\color{white},   
    basicstyle=\footnotesize\listingsttfamily,
    breakatwhitespace=false,         
    breaklines=true,                 
    belowcaptionskip=.5\baselineskip, 
    captionpos=b,                    
    commentstyle=\color{purple!40!black},
    deletekeywords={...},            
    escapeinside={\%*}{*)},          
    extendedchars=true,              
    identifierstyle=\color{blue},
    firstnumber=0,                
    stepnumber=1,                
    frame=trbl,
    keepspaces=true,                 
    keywordstyle=\bfseries\color{green!35!black},
    numbers=none,                    
    numbersep=10pt,                   
    numberstyle=\tiny, 
    rulecolor=\color{black},         
    showspaces=false,                
    showstringspaces=false,          
    showtabs=false,                  
    stepnumber=1,                    
    stringstyle=\color{orange},
    tabsize=2,	                   
    title=\lstname,                   
    xleftmargin=\parindent,
    }

\newcommand{\bijection}{{\normalfont\textsf{Bijection}}\xspace}

\newcommand{\Stack}{\mathcal{S}}
\newcommand{\scons}{\!::\!}
\newcommand{\spop}{\operatorname{pop}}
\newcommand{\spush}{\operatorname{push}}
\newcommand{\shead}{\operatorname{head}}

\newcommand{\stail}{\operatorname{tail}}
\newcommand{\sempty}{\operatorname{empty}}
\newcommand{\sreverse}{\operatorname{reverse}}

\newcommand{\twoway}{{\normalfont\textsf{Two-way}}\xspace}

\newcommand{\twowayb}{{\normalfont\textsf{Two-way bijection}}\xspace}
\newcommand{\twowaybs}{{\normalfont\textsf{Two-way bijections}}\xspace}

\newcommand{\onewayf}{{\normalfont\textsf{One-way function}}\xspace}
\newcommand{\onewayfs}{{\normalfont\textsf{One-way functions}}\xspace}

\newcommand{\honest}{{\normalfont\textsf{Honest}}\xspace}

\newcommand{\ForNo}{{\normalfont\textsf{ForNo}}\xspace}

\newcommand{\ICC}{{\normalfont\textsf{ICC}}\xspace}

\newcommand{\RBS}{{\normalfont\textsf{RBS}}\xspace}

\newcommand{\ForNoS}{{\normalfont\textsf{ForNo}}\xspace}

\newcommand{\FPTIME}{{\normalfont\textsf{FPTIME}}\xspace}

\newcommand{\NPTIME}{{\normalfont\textsf{NPTIME}}\xspace}

\newcommand{\PTIME}{{\normalfont\textsf{PTIME}}\xspace}

\newcommand{\TM}{{\normalfont\textsf{TM}}\xspace}

\newcommand{\lComma}{\lstiC{;}}

\newcommand{\blank}{{\normalfont \texttt{\char32}}}

\newcommand{\lstiC}[1]{\mbox{\normalfont\lstinline[style=c,mathescape=true]|#1|}}

\newcommand{\llangle}{\langle\!\langle}
\newcommand{\rrangle}{\rangle\!\rangle}

\newcommand\yield[1]{\mathrel{\stackrel{\makebox[0pt]{\mbox{\tiny #1}}}{\rightsquigarrow}}}
\usepackage{color}

\title{Towards a Characterization of Two-way Bijections \\
      in a Reversible Computational Model}
\author{Matteo Palazzo \and Luca Roversi}
\date{}

\begin{document}

\maketitle

\begin{center}
    \textit{Dipartimento di Informatica, Università di Torino} \\
    \texttt{matteo.palazzo@unito.it, luca.roversi@unito.it}
\end{center}

\begin{abstract}
We introduce an imperative, stack-based, and reversible computational model that characterizes \twowaybs both implicitly, concerning their computational complexity, and with zero-garbage.

\end{abstract}

\bigskip

\theoremstyle{plain}
\newtheorem{theorem}{Theorem}
\newtheorem{lemma}{Lemma}
\newtheorem{corollary}{Corollary}

\theoremstyle{definition}
\newtheorem{definition}{Definition}

\theoremstyle{remark}
\newtheorem{remark}{Remark}
\newtheorem{note}{Note}

\section{Introduction}
\label{section:new Introduction}

We start by recalling that if the class of \onewayfs exists \cite{arora2009computational,robshaw2011oneway}, each of its members would efficiently compute their output while making it computationally infeasible to reverse-engineer the corresponding input. This would mean that \onewayfs are a formal tool on which to base secure encryption \cite{goldreich2008foundations} and, in fact, showing the existence of \onewayfs would imply \PTIME $\neq$ \NPTIME. So, conversely, proving the non-existence of \onewayfs would dismantle many cryptographic assumptions, forcing to rethink how to achieve security.
\par
On the other hand, we recall that \twowaybs establish a domain where feasible invertibility holds. Consequently, any characterization of \twowaybs can help to delineate the space where \onewayfs \emph{should fail} to exist because a bijective function efficiently computable in both directions would inherently disqualify itself from being a \onewayf. So,

\paragraph{Motivations.}
Characterizing \twowaybs through a computational model could reveal structural properties of such bijections that extend beyond their mere definition. This perspective may help refine the class of functions where \onewayfs \emph{could} potentially emerge.

\paragraph{Contributions.}
We introduce \ForNo, an imperative, stack-based, and reversible computational model that we prove to characterize \twowaybs both implicitly and with zero-garbage.

By ``\emph{imperative and stack-based}'' we mean that \ForNo assigns a stack (of natural numbers)  to every of its registers.

By ``\emph{reversible}'' we mean that \ForNo is defined so that if any term $T$ produces a final configuration $c_f$, from an initial one $c_i$, then the inverse $-T$ of $T$ belongs to \ForNo as well and generates $c_i$ from $c_f$. Moreover, $-T$ is obtained from $T$ by a syntax-directed map.

By ``\emph{characterizes \twowaybs implicitly and with zero-garbage}'' we mean that \ForNo enjoys the following properties.
We prove that \ForNo is \FPTIME-sound and \FPTIME-complete.
Completeness ensures that every \FPTIME function can be computed in \ForNo. Soundness guarantees that if a function is computable in \ForNo then it is in \FPTIME.
We remark that \ForNo is \emph{implicitly} \FPTIME-sound because it follows Implicit Computational Complexity (\ICC) principles \cite{DalLago_2022}. That is, the terms in \ForNo adhere to syntactic constraints that inherently restrict computation length, eliminating reliance on ``external clocks'', unlike \PTIME Turing Machines.
As a corollary, \ForNo characterizes \twowaybs with zero-garbage for the two following reasons:
(i) \FPTIME-soundness and reversibility guarantees that if \ForNo can compute a bijection $f$ with zero-garbage, then both $f$ and $f^{-1}$ are in \FPTIME, which means  that $f$ is a \twoway\ \bijection;
(ii) \FPTIME-completeness guarantees that if $f$ is a \twoway \ \bijection, which implies that $f$ and $f^{-1}$ are in \FPTIME, then two terms $T_f$ and $T_{f^{-1}}$ in \ForNo exist that compute them, respectively.
Therefore, the Bennett Trick applied to $T_f$ and $T_{f^{-1}}$ allows us to build a term in \ForNo that computes $f$ with zero-garbage.

\paragraph{Related work.}
To the best of our knowledge, Kristiansen's reversible computational model \RBS in \cite{DBLP:journals/scp/Kristiansen22} is the first -- and possibly the only -- study linking \ICC with reversible computational models . However, \RBS characterizes \PTIME \emph{decision problems} and inherently \emph{reproduces part of the input as output}, meaning it is not zero-garbage. Thus, as far as we can see, \RBS cannot characterize the \twowaybs as seamlessly as \ForNo.

\paragraph{``Index''.}
Section~\ref{section:Normal iteration on notation} introduces \ForNo, its big-step operational semantics, and says why \ForNo is reversible.
Section~\ref{section:ForNo can compute all PTIME functions} and~\ref{section: ForNo is Correct w.r.t. FP} prove \ForNo \ \FPTIME-completeness and soundness.
Section~\ref{section:ForNo characterizes twoways bijections} proves that \ForNo characterizes \twowaybs.

\section{\textsf{ForNo} and its Operational Semantics}
\label{section:Normal iteration on notation}
Let us start defining \emph{Raw terms} as given by the the following grammar $\mathcal{T}$:
\begin{align*}
	T & ::= A \mid
                T \lComma T \mid
                \lstiC{IF}\  x \ \lstiC{=} \ n \ \{T\} \mid
                \lstiC{NORMAL} \ N \ \{S\}
&
    N & ::= y \mid y \lstiC{,} N
    \\
	A & ::=  \lstiC{SKIP} \mid
	\lstiC{PUSH[}n\lstiC{]} \ x
	\mid \lstiC{POP[}n\lstiC{]} \ x
    \\
    S & ::=  A
        \mid S \lstiC{;} S
	\mid \lstiC{IF}\  x \ \lstiC{=} \ n \ \{S\}
	\mid \lstiC{FOR}\ x\ C
	\mid \lstiC{ROF}\ x\ C
&
    C & ::= \{ S \} \mid \{ S \}\ C \enspace ,
\end{align*}
with $T$ its start symbol, and $n \in \mathbb{N}$. Registers are \lstiC{x}, \lstiC{y}, \lstiC{z} \ldots in a set $\mathcal{R}$, ranged over by meta-variables $x, y \ldots$.
From $A$ we get \emph{Atomic terms}.
The \emph{Sequential composition} is ``\lComma'', and \emph{selection} is \lstiC{IF}.
A \emph{normal block} is $\lstiC{NORMAL}\ N \ \{S\}$ with $N$ a `\lstiC{,}'-separated list of ``\emph{normal w.r.t. $S$\,}'' registers.
From $S$ we get \emph{safe terms}, the \emph{only ones} that can include \textit{forward/backward (finite) iterations} \lstiC{FOR} and \lstiC{ROF}, respectively.
Both \lstiC{FOR} and \lstiC{ROF} apply to a register and a list of \emph{bodies} $\{S\}$.
Finally,  we say ``\lstiC{PUSH}/\lstiC{POP} \emph{writes} $x$'',
``$x$ \emph{leads} \lstiC{IF}/\lstiC{FOR}/\lstiC{ROF}'', and
``$T$ and $S$ are a body'' in relation to selections, normal blocks, and forward/backward iterations.

\paragraph{Syntax of \ForNo.} \ForNo contains the Raw Terms generated from the rule $T$ subject to the following proviso:
(i) if $x$ leads a selection or an iteration with bodies $S_1, \ldots, S_n$, then $x$ must be \emph{read-only} in the bodies, that is no atomic term in $S_1, \ldots, S_n$ can write $x$;
(ii) registers that are normal w.r.t. a body $S$ must be read-only in $S$ and are the only ones that can lead iterations in $S$.

\begin{figure}[h!]
	\centering
	\begin{tabular}{c}
		\bottomAlignProof
		\AxiomC{}
		\RightLabel{\scriptsize \textsc{skip}}
		\UnaryInfC{$\langle \lstiC{SKIP}, \omega \rangle \Downarrow \omega$}
		\DisplayProof
        \qquad
        \bottomAlignProof
		\AxiomC{$\langle T, \omega \rangle \Downarrow \omega'$}
        \AxiomC{$\langle U, \omega' \rangle \Downarrow \omega''$}
		\RightLabel{\scriptsize \textsc{seq}}
		\BinaryInfC{$\langle T\lstiC{;}U, \omega \rangle \Downarrow \omega''$}
		\DisplayProof
		\\ \\
		\bottomAlignProof
		\AxiomC{$\spop_n(\phi(x), c) = (s, c')$}
		\RightLabel{\scriptsize \textsc{pop}}
		\UnaryInfC{$\langle \lstiC{POP[}n\lstiC{]} \ x, (\phi, c) \rangle \Downarrow (\phi[x\to s],c')$}
		\DisplayProof
		\quad
		\bottomAlignProof
		\AxiomC{$\spush_n(\phi(x), c) = (s, c')$}
		\RightLabel{\scriptsize \textsc{push}}
		\UnaryInfC{$\langle \lstiC{PUSH[}n\lstiC{]}\ x, (\phi, c) \rangle \Downarrow (\phi[x\to s],c')$}
		\DisplayProof
		\\ \\
		\bottomAlignProof
		\AxiomC{$\shead(\phi(x)) = n \!\!$}
		\AxiomC{$\!\! \langle P, (\phi,c) \rangle \Downarrow \omega'$}
		\RightLabel{\scriptsize \textsc{if-eq}}
		\BinaryInfC{$\langle \lstiC{IF}\ x \lstiC{=} n \ \{P\}, (\phi,c) \rangle \Downarrow \omega'$}
		\DisplayProof
		\qquad
		\bottomAlignProof
		\AxiomC{$\sempty(\phi(x)) \lor \shead(\phi(x)) \neq n$}
		\RightLabel{\scriptsize \textsc{if-neq}}
		\UnaryInfC{$\langle \lstiC{IF}\ x \lstiC{=} n \ \{P\}, (\phi,c) \rangle \Downarrow (\phi,c)$}
		\DisplayProof
		\\ \\
		\bottomAlignProof
		\AxiomC{$\Lbag\, \omega(x), [P_0, \ldots, P_n], \omega \,\Rbag\ \Downarrow \omega'$}
		\RightLabel{\scriptsize \textsc{for}}
		\UnaryInfC{$\langle \lstiC{FOR}\ x\ \{P_0\} \ldots \{P_n\}, \omega \rangle \Downarrow \omega'$}
		\DisplayProof
		\qquad
		\bottomAlignProof
		\AxiomC{$\Lbag \sreverse(\omega(x)), [P_0, \ldots, P_n], \omega \,\Rbag\ \Downarrow \omega'$}
		\RightLabel{\scriptsize \textsc{rof}}
		\UnaryInfC{$\langle \lstiC{ROF}\ x\ \{P_0\} \ldots \{P_n\}, \omega \rangle \Downarrow \omega'$}
		\DisplayProof
        \\ \\
		\bottomAlignProof
		\AxiomC{$\langle T, \omega \rangle \Downarrow \omega'$}
		\RightLabel{\scriptsize \textsc{n}}
		\UnaryInfC{$\langle \lstiC{NORMAL}\ N\ \{T\}, \omega \rangle \Downarrow \omega'$}
		\DisplayProof
		\qquad\qquad
        \bottomAlignProof
		\AxiomC{$\langle \lstiC{SKIP}, \omega \rangle \Downarrow \omega$}
		\RightLabel{\scriptsize \textsc{base}}
		\UnaryInfC{$\Lbag\, [\,], [P_0,\ldots,P_n], \omega \,\Rbag\ \Downarrow \omega$}
		\DisplayProof
        \\ \\
        \bottomAlignProof
		\AxiomC{$\langle P_{\operatorname{min}(i,n)}, \omega \rangle \Downarrow \omega''$}
		\AxiomC{$\Lbag\, t, [P_0,\ldots,P_n], \omega'' \,\Rbag\ \Downarrow \omega'$}
		\RightLabel{\scriptsize \textsc{step}}
		\BinaryInfC{$\Lbag\, i \scons t, [P_0,\ldots,P_n], \omega \,\Rbag\ \Downarrow \omega'$}
		\DisplayProof
    \end{tabular}
	\caption{Big-step semantics of \ForNoS.}
	\label{fig:semantic}
\end{figure}

\paragraph{Operational Semantics of \ForNo.}
Fig.~\ref{fig:semantic} introduces the deductive rules of a big-step operational semantics. The rules derive judgments $\langle L;\omega\rangle\Downarrow\omega'$, and $\Lbag s, L;\omega\Rbag\Downarrow\omega'$, where $L$ is a list of Row Terms\footnote{For example, $L$ can contain $\lstiC{FOR x \{SKIP\}}\in\mathcal{T}$ which is not in \ForNo because iterations can occur only in the body of \lstiC{NORMAL}.}, $s$ is a \emph{Stack} and $\omega, \omega'$ are \emph{States}.
\par
The set $\Stack$ of Stacks contains the empty stack $[ \, ]$, and every $n \scons t$ with $t\in\Stack$ and $n \in \mathbb{N}$, where $n$ is the \emph{top} of the stack. The length $|s|$ of $s\in\Stack$ counts the elements in $s$. The functions $\shead, \stail, \sempty$, $\sreverse$ are the expected ones.
\par
The set $\Omega$ of all States contains pairs $(\phi, c)$ where $\phi$ is a \emph{store} in the set $\Phi$ of \emph{Stores}, and a \emph{counter} $c \in \mathbb{N}$.
\par
On the one hand, every store maps registers to values in $\mathcal{S}$. For example, $\phi[\lstiC{x} \to 0 \scons 1 \scons [\,] \, ]$ is a store. It extends the store $\phi$ by binding the register $\lstiC{x}$ to the stack $0\scons1\scons [\,]$ whose top is $0$.
By $\varnothing$ we denote the store such that $\forall x, \varnothing(x) = [\,]$. For example, $\varnothing[\lstiC{x}\to 1::[\,]\,]$  evaluates $\lstiC{x}$ to $1::[\,]$ and every other register to $[\,]$.
On the other hand, the error counter ensures that regardless of how  $\spush_n$ and $\spop_n$, which push/pop the value $n$ on a stack, are interwoven
the resulting state can be undone, guaranteeing  a \emph{reversible} and \emph{total} operational semantics.
The counter $c$ is crucial because, in a state $(\phi, c)$, the value of $c$ is a new dimension which separates states, otherwise ``overlapping''\footnote{A full proof by means of the proof assistant \textsf{Coq} that $\spush_n$, and $\spop_n$ are each other inverses, that is $\spop_n(\spush_n(s,c)) = (s,c) = \spush_n(\spop_n(s,c))$, for every $n \in \mathbb{N}$ and $(s, c) \in \Stack \times \mathbb{N}$ is in \url{https://github.com/MatteoPalazzo/FORNO}.}.
For example, for any $\phi$ we have:
\begin{align}
\label{align:pop on two elements}
\langle\lstiC{POP[0]x},(\phi[\lstiC{x} \to 0 \scons 1 \scons [\,]\,],0)\rangle &\Downarrow (\phi[\lstiC{x} \to 1 \scons [\,]\,],0)
\\
\label{align:pop on one element}
\langle\lstiC{POP[0]x},(\phi[\lstiC{x} \to 1 \scons [\,]\,],0)\rangle &\Downarrow (\phi[\lstiC{x} \to 1 \scons [\,] \,],1) 
\enspace .
\end{align}
\noindent
In~\eqref{align:pop on two elements} the interpretation of \lstiC{POP[0]x} by $\spop_{0}$ successfully pops $0$ from the stack \lstiC{x}.  
In~\eqref{align:pop on one element} the interpretation of \lstiC{POP[0]x} by $\spop_{0}$ cannot pop $0$ from the stack \lstiC{x} because its top is not $0$. Without the counter, both 
\eqref{align:pop on two elements} and~\eqref{align:pop on one element} would generate an identical final store, starting from two different stores, breaking reversibility.
We classify every state $(\phi,0)\in \Omega$ as \emph{sound}. 
For every $\omega = (\phi, c) \in \Omega$, we write $\omega(x)$ meaning $\phi(x)$.
\par
Some final comments on Fig.~\ref{fig:semantic} are worth doing to completely describe the operational semantics. 
The rules \textsc{if-eq}, and \textsc{if-neq} say whether the body should be executed or not, depending on the value on top of the stack held by $x$.
The rule \textsc{for} unfolds the iteration it represents. It appeals to \textsc{base}, and \textsc{step} which stepwise chose the term in $[P_0,\ldots,P_n]$ to interpret until the stack $\omega(x)$, or $\sreverse(\omega(x))$, empties.
Let $i$ be the current stack head, \textsc{step} interprets $P_i$ if $i \leq n$, and $P_n$ whenever $i > n$. So, $P_n$ is a sort of ``default choice''.
The unfolding lasts exactly $|x|$ steps, implementing an iteration \emph{on notation}. The rule \textsc{rof} is \textsc{for} after reversing $\omega(x)$.
We remark reversing $\omega(x)$ is local to \textsc{rof}.
The rule \textsc{n} delimits a scope. 

\paragraph{\ForNo is reversible.} The map $-(\cdot): \mathcal{T} \to \mathcal{T}$ in Fig.~\ref{fig:negation of terms}, inductively defined on the structure of its argument $T$, yields  
$-T$, which we call \emph{inverse of $T$}. The reason is that we can prove that composing $T$ with $-T$, or vice versa, is interpreted as \lstiC{SKIP}, that is the identity, by the operational semantics in Fig.~\ref{fig:semantic}. The key step of the proof is to show that $\langle T\lstiC{;} -T, (\phi,c) \rangle \Downarrow (\phi',c')$, with $c = c'$ and $\phi (x) = \phi'(x)$, for every $x, T$.

\begin{figure}
$
\begin{array}{c}
-(\lstiC{SKIP}) = \lstiC{SKIP}
\qquad\qquad\qquad \qquad 
-(Q_0 \lstiC{;} Q_1) = -(Q_1) \lstiC{;} -(Q_0)
\\[2pt]
-(\lstiC{PUSH[}n{]} \ x) = \lstiC{POP[}n{]} \ x 
\qquad 
\quad \quad 
-(\lstiC{POP[}n{]} \ x) = \lstiC{PUSH[}n{]} \ x 
\\[2pt]
-(\lstiC{IF} \ x\lstiC{=}n \ \lstiC{\{}Q\lstiC{\}}) 
=  \lstiC{IF} \ x\lstiC{=}n \ \lstiC{\{}-(Q)\lstiC{\}} 
\quad\quad 
-(\lstiC{NORMAL} \ V \lstiC{\{}Q\lstiC{\}}) 
= \lstiC{NORMAL} \ V \lstiC{\{}-(Q)\lstiC{\}} 
\\[2pt]
-(\lstiC{FOR} \ x \ \lstiC{\{}Q_0\lstiC{\}} \ldots {\{}Q_n\lstiC{\}}) 
=   \lstiC{ROF} \ x \ \lstiC{\{}-(Q_0)\lstiC{\}} \ldots {\{}-(Q_n)\lstiC{\}} 
\\[2pt]
-(\lstiC{ROF} \ x \ \lstiC{\{}Q_0\lstiC{\}} \ldots {\{}Q_n\lstiC{\}}) 
=   \lstiC{FOR} \ x \ \lstiC{\{}-(Q_0)\lstiC{\}} \ldots {\{}-(Q_n)\lstiC{\}}
\end{array}
$
\caption{The map $-(\cdot): \mathcal{T} \to \mathcal{T}$ to produce the inverse $-T$ of $T$.}
\label{fig:negation of terms}
\end{figure}


\section{\ForNo is \FPTIME-complete}
\label{section:ForNo can compute all PTIME functions}
``\ForNo is \FPTIME-complete'' means that \ForNo computes \emph{all} the \FPTIME functions. To prove it we define a \emph{compiler} $\llbracket \cdot \rrbracket$ mapping every polynomial time Turing machine $M$ to a term in \ForNo that simulates $M$.
We define $\llbracket \cdot \rrbracket$ by modifying an analogous compiler from \cite{Kristiansen2004OnTC} to conform to the syntactic constraints of \ForNo. 

\paragraph{Preliminaries.} We will refer to Turing machines (\TM) as a tuple composed by:
(i) a set $Q$ of states, with $q_0$ as initial and $q_\textsc{halt}$ as halting states;
(ii) an input/output alphabet $\Sigma$, s.t. $\blank \notin \Sigma$;
(iii) a semi-infinite tape with tape alphabet $\Gamma \supset \Sigma$ s.t. $\blank \in \Gamma$.
(iv) a transition function 
$\delta: (Q \setminus \{q_\textsc{halt}\} \times \Gamma \to Q \times \Gamma \times \{L,R\})$. 
A \emph{configuration} $c$ of a \TM is a triple $(u,q,a\cdot v)$ where $q$ is the current state, $a$ is the read character, $u$ and $v$ are two strings representing the left and the right-hand parts of the tape, respectively. 
More precisely, $v$ captures only the meaningful portion to the right of the head, excluding any occurrences of \blank$\xspace$ at the tape's tail. Consequently, $v$ never terminates with a \blank.
For every configuration $c$ and $c'$ we say `$c$ \emph{yields} $c'$', and write $c \rightsquigarrow c'$, to denote a single-step transition from $c$ to $c'$, while $c \yield{n} c'$, with $n \geq 1$, stands for an $n$-steps transition.
\par
A \TM $M$ is a \emph{Polynomial Time} \TM (\PTIME\ \TM) if and only if there exists a polynomial $p$ that characterizes $M$, that is, such that $M$ enters $q_\textsc{halt}$ after at most $p(|w|)$ steps, for \emph{every} finite input $w\in \Sigma^*$. We write $M_p$ to denote a \PTIME\ \TM with a characterizing polynomial $p$. Without loss of generality, we assume that $p(x) = ax^b$ for some $a,b \in \mathbb{Z}^+$\footnote{$\mathbb{Z}^+$ is the set of positive natural numbers.}. A total function $f: \Sigma^* \to \Sigma^*$ is in \FPTIME 
if and only there exists a \PTIME\ \TM $M_f$ such that: (i) $M$ enters $q_\textsc{halt}$ with \emph{only} $f(w)$ on its tape, whenever $M$ starts from $q_0$ with $w$ on its tape; (ii) the head of $M$ is over the leftmost cell of the tape, when in state $q_0$ or $q_\textsc{halt}$.

\begin{definition}
\label{def:forno_computable_function}
For any alphabet $\Sigma$, we call any bijection $\llangle \cdot \rrangle: \Sigma \to \{0,\ldots,|\Sigma|-1\}$ an \emph{encoding} function. We overload the notation by saying that an encoding $\llangle x \rrangle$ of a string $x \in \Sigma^*$ is a stack  defined as $\llangle \epsilon \rrangle = [\,]$ and $\llangle a \cdot x \rrangle = \llangle a \rrangle \scons \llangle x \rrangle$\footnote{By $a \cdot x$ we mean the string obtained by composing the character $a$ with $x$}.

Let $f: \Sigma^* \to \Sigma^*$ be a total function, $T$ be a term in \ForNo, both $\operatorname{in}$ and $\operatorname{out}$ be two (not necessarily distinct) registers of $T$, and $\llangle \cdot \rrangle$ be an encoding function. We say `\,$T$ \emph{computes} $f$' if and only if, 
$\phi$ exist s.t.
$\langle T, (\varnothing[\operatorname{in} \to \llangle x \rrangle],0)\rangle \Downarrow (\phi,0)$ with $\phi(\operatorname{out}) = \llangle f(x) \rrangle$. Moreover, we say `\,$T$ \emph{computes} $f$ \emph{with zero-garbage}' if and only if $\phi(y) = [\,]$, for every $y \neq \operatorname{out}$.
\end{definition}

\begin{wrapfigure}{l}{0.45\textwidth}
    \vspace{-20pt}
    \hspace{5pt}
  \begin{minipage}{0.4\textwidth}
\begin{lstlisting}[style=C, caption={Structure of $\llbracket M_p \rrbracket$.}, label=lst:tm_forno_compiler, mathescape = true, numbers = left]
POLYNOMIAL$_p$; $\label{line:polynomial}$
PUSH[$\llangle q_0 \rrangle$] q;
NORMAL p { $\label{line:mainStart}$
  FOR p {SIMULATE}}; $\label{line:mainEnd}$
REMOVE-BLANKS  $\label{line:sanitize}$\end{lstlisting}
  \end{minipage}
  \vspace{-20pt}
\end{wrapfigure}
\paragraph{The compiler $\llbracket \cdot \rrbracket$.} Let $M_p$ be a \PTIME\ \TM. Listing \ref{lst:tm_forno_compiler} illustrates the structure of $\llbracket M_p \rrbracket$
where \lstiC{POLYNOMIAL$_p$}, \lstiC{SIMULATE} and \lstiC{REMOVE-BLANKS} are \emph{macros} to be replaced by their definition.

We require the encoding $\llangle \cdot \rrangle$ to satisfy some constraints.
Let $\Gamma = \{a_0,\ldots, a_{n-1}\}$ and $\Sigma = \{b_0,\ldots,b_{m-1}\}$, with $n>m$, the alphabets of $M_p$, then (i) $\llangle a_i \rrangle = i$ for each $i \in \{0, \ldots, n-1\}$; (ii) $\llangle a \rrangle > \llangle b \rrangle$, for every $a \in \Gamma \setminus \Sigma$ and $b \in \Sigma$. 

A configuration $c = (u,q,v)$ of $M_p$ is represented in $\llbracket M_p \rrbracket$ by the registers $\lstiC{lft}$, $\lstiC{q}$ and $\lstiC{rgt}$, each standing for the the corresponding element of $c$. Formally, $\sigma$ \emph{simulates} $c$, written $\sigma \approx c$, if: (i) $\sigma$ is sound; (ii) $\sigma(\lstiC{lft}) = \llangle u \rrangle$; (iii) $head(\sigma(\lstiC{q})) = \llangle q \rrangle$ and $\sigma(\lstiC{rgt}) = \llangle v \rrangle +\!\!+ \; bs$ where $bs$ is a possibly empty stack of $\llangle \blank \rrangle$ and $+\!+$ is concatenation. If $bs = [\,]$, then we say `$\sigma$ \emph{cleanly simulates} $c$' and write $\sigma \cong c$. 

The simulation of $M_p$ on a given input $w$ by $\llbracket M_p \rrbracket$ assumes that $\llbracket M_p \rrbracket$ begins in the state $\sigma = (\varnothing[\lstiC{rgt} \to \llangle w \rrangle],0)$ and operates through three phases.

Firstly (see Line \ref{line:polynomial} of Listing \ref{lst:tm_forno_compiler}), $\llbracket M_p \rrbracket$ stores inside \lstiC{p} a stack whose length is the upper-bound $p(|w|) = c|w|^n$, for some $c$ and $n$, on the number of steps that $M_p$ does.
The term $\lstiC{POW}^n$ increments the length of \lstiC{p} by $|w|^n$ and is defined as $
\lstiC{NORMAL} \ \lstiC{rgt} \ \lstiC{\{}\underbrace{\lstiC{FOR} \ \lstiC{rgt} \ \lstiC{\{FOR} \ \lstiC{rgt} \ \lstiC{\{} \ \ldots \ \lstiC{FOR} \ \lstiC{rgt} \lstiC{\{}}_{\text{$n$ times}}\lstiC{PUSH[1]} \ \lstiC{p} \lstiC{\}} \ldots \lstiC{\}\}\}\}}$.
Therefore, $\lstiC{POLYNOMIAL}_p = \underbrace{\lstiC{POW}^n\lComma\ldots\lComma\lstiC{POW}^n}_{\text{$c$ times}}$ increments the length of $\lstiC{p}$ by $c|w|^n$. 

\begin{tabular}{l l}
& \\
    \begin{minipage}{.45\textwidth}
        \begin{lstlisting}[style=C, caption={\lstiC{TRANSITION}$_L(q', a')$.}, label=lst:transitionL, mathescape = true, numbers = left]
PUSH[$q'$] q; $\label{lstline:setStateL}$   
TO(rgt,g); $\label{lstline:set0AL}$
PUSH[$a'$] rgt; $\label{lstline:setA1AL}$
EMPTY(lft); $\label{lstline:emptyL}$
IF empty = 1 {
 TO(lft,rgt)} $\label{lstline:toLftRgt}$
\end{lstlisting}  
    \end{minipage} &  
    \begin{minipage}{.47\textwidth}
        \begin{lstlisting}[style=C, caption={\lstiC{TRANSITION}$_R(q', a')$.}, label=lst:transitionR, mathescape = true, numbers= left]
PUSH[$q'$] q; $\label{lstline:setStateR}$
TO(rgt,g); $\label{lstline:set0AR}$
PUSH[$a'$] rgt; $\label{lstline:setA1AR}$
TO(rgt,lft); $\label{lstline:toRgtLft}$
EMPTY(rgt); $\label{lstline:emptyR}$
IF empty = 0 {PUSH[$\llangle \blank \rrangle$] rgt} $\label{lstline:pushBlank}$ 
\end{lstlisting}  
    \end{minipage}
\end{tabular}

\begin{lstlisting}[style=C, caption={\lstiC{SIMULATE}. In italic we specify how to build the term.}, label=lst:simulate, mathescape = true, numbers = left]
$\textit{for each } q \in Q \setminus \{q_{\textsc{HALT}}\}:$ IF q = $\llangle q \rrangle$ {PUSH[$\llangle q \rrangle$] qStart};
$\textit{for each } a \in \Gamma$: IF rgt = $\llangle a \rrangle$ {PUSH[$\llangle a \rrangle$] aStart};
$\textit{for each } q \in Q \setminus \{q_{\textsc{HALT}}\}:$
 IF qStart = $\llangle q \rrangle$ { 
  $\textit{for each } a \in \Gamma$:
   IF aStart = $\llangle a \rrangle$ {
    $\textit{let } \delta(q,a) = (q',a',d)$ in: TRANSITION$_d$($\llangle q' \rrangle$, $\llangle a' \rrangle$)}};
\end{lstlisting}

Secondly (see Lines \ref{line:mainStart} and \ref{line:mainEnd} of Listing \ref{lst:tm_forno_compiler}), $\llbracket M_p \rrbracket$ iterates \lstiC{SIMULATE} in Listing \ref{lst:simulate} as many times as \lstiC{p}, that is $p(|w|)$ times. The purpose of \lstiC{SIMULATE} is to simulate a step of $M_p$.  By iterating $p(|w|)$ times, we simulate all the steps of $M_p$.
If the currently simulated configuration is not an halting one, \lstiC{SIMULATE} selects and simulate the appropriate transition through $\lstiC{TRANSITION}_d$. Such terms, defined in Listings \ref{lst:transitionL} and \ref{lst:transitionR}, simulate a transition by first changing the current state, then by changing the scanned character and finally by moving the simulated head to the left or right.   Due to space restrictions, we omit the full definition of \lstiC{EMPTY} and \lstiC{TO}. One checks if a stack is empty (setting the top of \lstiC{empty} to $0$, otherwise to $1$), the other moves the top from a stack to another one, respectively. Finally, \lstiC{g} and \lstiC{g1} are used to store the garbage of the computation.

\begin{theorem}\label{theorem:simulate-base}
Let $M$ be a $\PTIME \ \TM$. For every $\sigma$ and $c$ s.t. $\sigma \approx c$:
\begin{enumerate}
    \item if $c$ is an halting configuration, then $\langle \lstiC{SIMULATE}, \sigma \rangle \!\! \Downarrow \!\! \tau$, for some $\tau \approx c$. Otherwise, if it exists $c'$ s.t. $c \rightsquigarrow c'$, then $\langle \lstiC{SIMULATE}, \sigma \rangle \Downarrow \tau$, for some $\tau \approx c'$.
    \item If $c \yield{p} c_{\textsc{halt}}$ for an halting configuration $c_{\textsc{halt}}$, then for any $b\in\mathcal{S}$ s.t. $|b| \geq p$, we have $\Lbag b, [\lstiC{SIMULATE}], \sigma \Rbag \Downarrow \tau$, for some $\tau \approx c_{\textsc{halt}}$.
\end{enumerate}
\end{theorem}

\begin{lstlisting}[style=C, caption={The macro \lstiC{REMOVE-BLANK}.}, label=lst:sanitize, mathescape = true, numbers = left]
NORMAL rgt {ROF rgt {PUSH[$\llangle a_0 \rrangle$] g1} ... {PUSH[$\llangle a_{n-1} \rrangle$] g1}}; $\label{lstline:rgt_copy}$
NORMAL g1 {
  FOR g1 {POP[$\llangle a_0 \rrangle$] rgt} ... {POP[$\llangle a_{n-1} \rrangle$] rgt}; $\label{lstline:rgt_erase}$
  ROF g1 {PUSH[$\llangle b_0 \rrangle$] rgt} ... {PUSH[$\llangle b_{m-1} \rrangle$] rgt} {SKIP}} $\label{lstline:g1_copy}$
\end{lstlisting} 

Finally (see Line \ref{line:sanitize} of Listing \ref{lst:tm_forno_compiler}), $\llbracket M_p \rrbracket$ cleans the output eliminating the spurious occurrences of $\blank\xspace$ that are added to \lstiC{rgt} when scanning the empty segment of the tape (see Lines \ref{lstline:emptyR} and \ref{lstline:pushBlank} of Listing \ref{lst:transitionR}).
In fact, if $M_p$ halts in $c_{\textsc{halt}} = (\epsilon,q_{\textsc{halt}},v)$ for some $v$, then after the main iteration 
(see Lines \ref{line:mainStart} and \ref{line:mainEnd} of Listing \ref{lst:tm_forno_compiler}) 
$\llbracket M_p \rrbracket$ halts in a state $\tau \approx c_{\textsc{halt}}$ s.t. $\tau(\lstiC{rgt}) = \llangle v \rrangle +\!\!+ bs$, where $bs$ is a stack of \blank. \lstiC{REMOVE-BLANK} in Listing \ref{lst:sanitize} removes $bs$ from \lstiC{rgt} leading to a final state $\tau'$ s.t. $\tau'(\lstiC{rgt}) = \llangle v \rrangle$, hence $\tau' \cong c_{\textsc{halt}}$. The conclusion is:

\begin{theorem}\label{theorem:soundness-simulation}
Let $M$ a \PTIME \ \TM. For every input $w$ of $M$, let $c_w = (\epsilon, q_0, w)$ be the starting configuration of $M$ and $c_{\textsc{halt}}$ the configuration in which $M$ halts.
Then $\langle \llbracket M \rrbracket , (\varnothing[\lstiC{rgt} \to \llangle w \rrangle],0) \rangle \Downarrow \tau$ s.t. $\tau \cong c_{\textsc{halt}}$.
\end{theorem}

\noindent By definition of \FPTIME, the direct consequence of Thm. \ref{theorem:soundness-simulation} is:

\begin{corollary}\label{theorem:completenessPTM short}
If a function $f$ is in \FPTIME, then $f$ is computable in \ForNo.
\end{corollary}


\section{\ForNo is \FPTIME-sound}
\label{section: ForNo is Correct w.r.t. FP}

``\ForNo is \FPTIME-sound'' means that all its computations terminate within a polynomial number of steps dependent on the initial configuration's dimension. Formally:
\begin{theorem}\label{theorem:forno_fptime_sound}
For every term $T$  in $\ForNo$, if $T$ computes  the function $f$, then $f$ is in \FPTIME.
\end{theorem}
\noindent
The theorem holds because \ForNo is designed to meet two criteria identified in \cite{DBLP:journals/jcss/KasaiA80,Kristiansen2004OnTC} as sufficient for a language to be in \FPTIME. We summarize them.
\par
The first one is that, according to \cite{Kristiansen2004OnTC}, \ForNo is a \emph{stack language} because it is an imperative computational model such that:
(i) variables store stacks over some alphabet;
(ii) programs are obtained as the least set closed under sequential composition, conditional and iteration, starting from base commands like $\spop$ and $\spush$, for example, which operate on variables;
(iii) conditional selection depends on the top element of its stack argument;
(iv) iterations proceed according to stack-based notation, lasting exactly as many steps as the length of the stack in the variable driving the iteration unfolding.
\par
The second one is that, the syntactical constraints on every instance of $\lstiC{NORMAL}$ $N \ \{S\} $ implies that \ForNo contains \emph{simple terms} only, where, for a term to be simple means that it does not include iterations whose body $T$ contains a variable $x$ whose value length depends on itself in $T$, which would mean that the length of $x$ can grow exponentially.
\par
Therefore, from \cite{DBLP:journals/jcss/KasaiA80,Kristiansen2004OnTC} we can conclude that the execution of a simple term of a stack language always concludes within a polynomial number of steps depending on the input length.

For example, \ForNo rules out terms like:
\begin{align}
\label{align:term rules out}
&\lstiC{NORMAL x,y,z \{FOR z \{FOR x \{PUSH[1] y\};FOR y \{PUSH[1] x\}\}\}}
\enspace .
\end{align}
\noindent
In~\eqref{align:term rules out}, the term \lstiC{\{FOR x \{PUSH[1] y\};FOR y \{PUSH[1] x\}} implies that the value of \lstiC{x} controls the value of \lstiC{y}, and vice versa.
Consequently, \lstiC{x} controls the growth of its length.
Nesting the self-control of the growth of \lstiC{x} inside the iteration led by \lstiC{z} would imply an exponential blow-up of the dimension of the stack in \lstiC{x}: at each iteration the length of \lstiC{x} grows, which causes the length of \lstiC{x} to grow even more at the next iteration. This would require an exponential time, just to write the stack. But, we exclude this term because \lstiC{x} and \lstiC{y} must be read-only.

%


\begin{figure}
\centering
\begin{tikzpicture}[scale=1]
\node[draw = none](h1) at (0,0) {$f$ is a \twowayb};

\node[draw = none](h21) at (-2.5,-0.75) {$f \in \FPTIME$};

\node[draw = none](h22) at (2.5,-0.75) {$f^{-1} \in \FPTIME$};

\draw[->] (h1) -- (h21);
\draw[->] (h1) -- (h22);

\node[draw = none] at (6.5,-0.75) {(by Definition)};

\node[draw, rectangle, minimum width=1cm, minimum height=1cm] (tf) at (-2.5,-2.25) {$T_f$};
\node at ([xshift=-1cm]tf) {$\llangle x \rrangle$};
\node at ([xshift=1.25cm, yshift = 0.35cm]tf) {$\llangle f(x) \rrangle$};
\node at ([xshift=1.25cm, yshift = -0.35cm]tf) {$g(x)$};

\node[draw, rectangle, minimum width=1cm, minimum height=1cm] (tf-inv) at (2.5,-2.25) {$T_{f^{-1}}$};
\node at ([xshift=-1.25cm]tf-inv) {$\llangle f(x) \rrangle$};
\node at ([xshift=1.25cm, yshift = 0.35cm]tf-inv) {$\llangle x \rrangle $};
\node at ([xshift=1.25cm, yshift = -0.35cm]tf-inv) {$ g'(x)$};

\draw[->] (h21) -- (tf);
\draw[->] (h22) -- (tf-inv);

\node[draw, rectangle, minimum width=1cm, minimum height=1cm] (btf) at (-2.5,-4.25) {$B[T_f]$};
\node at ([xshift=-1.35cm]btf) {$\llangle x \rrangle$};
\node at ([xshift=1.25cm, yshift = 0.35cm]btf) {$\llangle f(x) \rrangle$};
\node at ([xshift=1.25cm, yshift = -0.35cm]btf) {$\llangle x \rrangle$};
\node (fxbtf) at ([xshift=1.70cm, yshift = 0.35cm]btf){};
\node (xbtf) at ([xshift=1.70cm, yshift = -0.35cm]btf){};

\node[draw, rectangle, minimum width=1cm, minimum height=1cm] (btf-inv) at (2.5,-4.25) {$-B[T_{f^{-1}}]$};
\node at ([xshift=1.70cm]btf-inv) {$\llangle f(x) \rrangle$};
\node  at ([xshift=-1.5cm, yshift = 0.35cm]btf-inv) {$\llangle x \rrangle$};
\node at ([xshift=-1.5cm, yshift = -0.35cm]btf-inv) {$\llangle f(x) \rrangle$};
\node (xbtf-inv) at ([xshift=-1.9cm, yshift = 0.35cm]btf-inv) {};
\node (fxbtf-inv) at ([xshift=-1.9cm, yshift = -0.35cm]btf-inv) {};

\node[draw = none] at (6.5,-2.25) {(by Corollary \ref{theorem:completenessPTM short})};

\node[draw, dashed, rectangle, fit=(btf)(btf-inv), inner sep=0.25cm, minimum width=6.75cm] (group) {};
\node at ([yshift=1cm]group) {$R_f$};

\draw[->] (tf) -- (btf);
\draw[->] (tf-inv) -- (btf-inv);

\draw[->] (xbtf) -- (xbtf-inv);
\draw[->] (fxbtf) -- (fxbtf-inv);

\node[draw = none, align=left] at (6.5,-4.25) {(by ``Bennett trick'' \\ and inversion)};

\end{tikzpicture}
\caption{Proof scheme of the ``Only if direction'' of Thm.~\ref{theorem:forno_characterize_two-way}. $T_f$ computes $f(x)$, producing $g(x)$ as garbage. $B[T_f]$ results from applying the ``Bennett Trick'' to $T_f$. Analogously for $T_{f^{-1}}$. We remark a strong analogy with \cite[Fig. 1]{Axelsen2011WhatDR}.}
\label{fig:two-way-proof-scheme}
\end{figure}
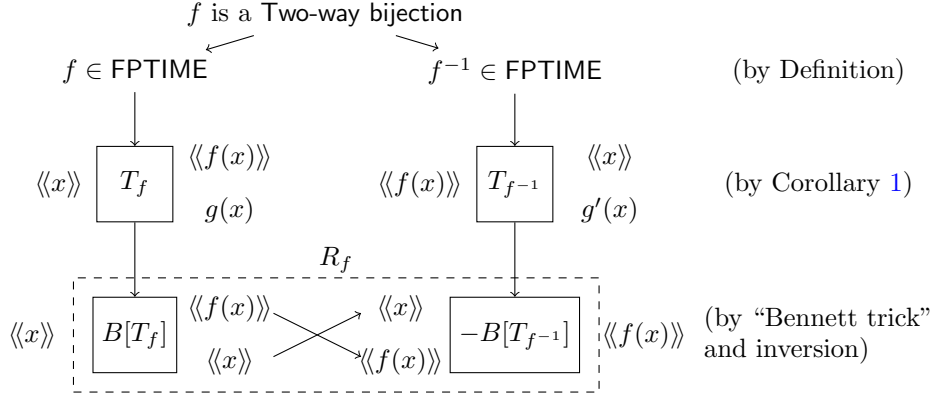

\section{\ForNo Characterizes {\normalfont\textsf{Two-way Bijections}}\xspace}
\label{section:ForNo characterizes twoways bijections}

From \cite{TheComplexityTheoryCompanion:02} we recall that a bijection $f$ is: (i)
\honest if $\forall x, |x| \leq q(|f(x)|)$ for some polynomial $q$, that is $|x|$ is at most polynomially longer than $|f(x)|$;
(ii) a \twowayb if both $f$ and $f^{-1}$ are \honest and in \FPTIME.

\begin{theorem}\label{theorem:forno_characterize_two-way}
A bijection $f: \Sigma^* \to \Sigma^*$ is a \twowayb  if and only if \ForNo computes $f$ with zero-garbage.
\end{theorem}
\begin{proof}
\emph{``If direction''.}
Let $T_f \in \ForNo$ computes $f$ with zero-garbage. Therefore, due to the reversibility, $-T_{f}$ computes $f^{-1}$.
Thm.~\ref{theorem:forno_fptime_sound} implies that $f, f^{-1}$ are in \FPTIME.
Moreover, $f \in \honest$ because:
(i) by definition, $\forall y, |f^{-1}(y)| \leq p(|y|)$, for some polynomial $p$,
otherwise $f^{-1}$ would not be in \FPTIME;
(ii) since $f$ is a bijection, replacing $f(x)$  for $y$ we get $\forall x, |f^{-1}(f(x))| \leq p(|f(x)|)$, meaning that $\forall x, |x| \leq p(|f(x)|)$.
Of course, $f^{-1} \in \honest$ for similar reasons.
\par
\emph{``Only if direction''.}
We read Fig. \ref{fig:two-way-proof-scheme} to justify why it holds.
Since $f$ is a \twowayb both $f, f^{-1} \in \FPTIME$.
By Corollary.~\ref{theorem:completenessPTM short} both $T_f$ and $T_{f^{-1}}$ exist in \ForNo that compute $f$ and $f^{-1}$, respectively. Without loss of generality we assume that
(i) \lstiC{x} is the input/output register of $T_f$;
(ii) \lstiC{fx} is the input/output register of $T_{f^{-1}}$;
(iii) \lstiC{fx} is not in $T_f$ and \lstiC{x} is not in $T_{f^{-1}}$; (iv) $T_f$ and $T_{f^{-1}}$ use the same encoding function $\llangle \cdot \rrangle$. Then both
$\langle T_f, (\varnothing[\lstiC{x} \to \llangle x \rrangle],0) \rangle \Downarrow \tau$ such that $\tau(\lstiC{x}) = \llangle f(x) \rrangle$ and
$\langle T_{f^{-1}}, (\varnothing[\lstiC{fx} \to \llangle f(x) \rrangle],0) \rangle \Downarrow \omega$ such that $\omega(\lstiC{fx}) = \llangle x \rrangle$, where both $\tau$ and $\omega$ are sound state.
\par
By using the ``Bennett trick'' we define $B[T_f] = T_f \lstiC{;COPY(x;fx);}-T_f$ and $B[T_{f^{-1}}] = T_{f^{-1}} \lstiC{;COPY(fx;x);}-T_{f^{-1}}$ such that both
$\langle B[T_f] , (\varnothing[\lstiC{x} \to \llangle x \rrangle],0)\rangle \Downarrow (\varnothing[\lstiC{x} \to \llangle x \rrangle,\lstiC{fx} \to \llangle f(x) \rrangle],0)$ and
$\langle B[T_{f^{-1}}] , (\varnothing[\lstiC{fx} \to \llangle f(x) \rrangle],0)\rangle \Downarrow (\varnothing[\lstiC{fx} \to \llangle f(x) \rrangle, \lstiC{x} \to \llangle x \rrangle],0)$.
The term $\lstiC{COPY}(x,y)$ which copies $x$ inside $y$, is  defined as $\lstiC{NORMAL} \ x \ \lstiC{\{ROF} \ x \ \lstiC{\{PUSH[}\llangle a_0 \rrangle\lstiC{]} \ y \lstiC{\}\}} \ldots \lstiC{\{PUSH[}\llangle a_n \rrangle\lstiC{]} \ y \lstiC{\}}$, supposing $\Sigma = \{a_0,\ldots, a_n\}$.
\par
Finally, by sequentially composing $B[T_f]$ and the inverse of $B[T_{f^{-1}}]$ we obtain the term $R_f = B[T_f]\lstiC{;}-B[T_{f^{-1}}]$ s.t.: $\langle R_f , (\varnothing[\lstiC{x} \to \llangle x \rrangle],0)\rangle \Downarrow (\varnothing[\lstiC{fx} \to \llangle f(x) \rrangle],0)$, which by definition \emph{computes $f$ with zero-garbage}.
\end{proof}


\end{document}